\documentclass[letter,11pt]{article}

\usepackage{setspace}
\usepackage{subcaption}
\usepackage{enumerate}
\usepackage{amsmath}
\usepackage{amsfonts}
\usepackage{graphicx}
\usepackage{amssymb}
\usepackage{multirow}
\usepackage{geometry}
\usepackage{soul}
\usepackage{colortbl}
\usepackage{wrapfig}
\usepackage{algorithm}
\usepackage{algorithmicx}
\usepackage{algpseudocode}
\usepackage{amsthm}
\usepackage{todonotes}
\usepackage{thmtools} 
\usepackage{thm-restate}

\algtext*{EndWhile}
\algtext*{EndIf}
\algtext*{EndFor}

\newtheorem{lemma}{Lemma}
\newtheorem{lemma*}[lemma]{Lemma*}
\numberwithin{lemma}{section}
\newtheorem{theorem}[lemma]{Theorem}
\newtheorem{theorem*}[lemma]{Theorem*}

\newtheorem{definition}[lemma]{Definition}

\newtheorem{fact}[lemma]{Fact}

\usepackage{wrapfig}
\usepackage{bm}

\usepackage{booktabs}
\usepackage[pdftex, plainpages = false, pdfpagelabels, 
                 bookmarks=false,
                 bookmarksopen = true,
                 bookmarksnumbered = true,
                 breaklinks = true,
                 linktocpage,
                 pagebackref,
                 colorlinks = true,  
                 linkcolor = blue,
                 urlcolor  = blue,
                 citecolor = red,
                 anchorcolor = green,
                 hyperindex = true,
                 hyperfigures
                 ]{hyperref}

\title{\LARGE Minimum-Membership Geometric Set Cover, Revisited}
\author{\normalsize Sayan Bandyapadhyay\thanks{Portland State University, USA, \texttt{sayan.bandyapadhyay@gmail.com}.}
\and \normalsize William Lochet\thanks{LIRMM, Université de Montpellier, CNRS, Montpellier, France, \texttt{william.lochet@gmail.com}.}
\and \normalsize Saket Saurabh\thanks{Institute of Mathematical Sciences, Chennai, India, \texttt{saket@imsc.res.in}.}
\and \normalsize Jie Xue\thanks{New York University Shanghai, China, \texttt{jiexue@nyu.edu}.}}
\date{}

\begin{document}


\maketitle

\begin{abstract}
We revisit a natural variant of the geometric set cover problem, called \textit{minimum-membership geometric set cover} (MMGSC).
In this problem, the input consists of a set $S$ of points and a set $\mathcal{R}$ of geometric objects, and the goal is to find a subset $\mathcal{R}^* \subseteq \mathcal{R}$ to cover all points in $S$ such that the \textit{membership} of $S$ with respect to $\mathcal{R}^*$, denoted by $\mathsf{memb}(S,\mathcal{R}^*)$, is minimized, where $\mathsf{memb}(S,\mathcal{R}^*) = \max_{p \in S} |\{R \in \mathcal{R}^*: p \in R\}|$.
We give the first polynomial-time approximation algorithms for MMGSC in $\mathbb{R}^2$.
Specifically, we achieve the following two main results.
\begin{itemize}
    \item We give the first polynomial-time constant-approximation algorithm for MMGSC with unit squares.
    This answers a question left open since the work of Erlebach and Leeuwen [SODA'08], who gave a constant-approximation algorithm with running time $n^{O(\mathsf{opt})}$ where $\mathsf{opt}$ is the optimum of the problem (i.e., the minimum membership).
    \item We give the first polynomial-time approximation scheme (PTAS) for MMGSC with halfplanes.
    Prior to this work, it was even unknown whether the problem can be approximated with a factor of $o(\log n)$ in polynomial time, while it is well-known that the minimum-size set cover problem with halfplanes can be solved in polynomial time.
\end{itemize}
We also consider a problem closely related to MMGSC, called \textit{minimum-ply geometric set cover} (MPGSC), in which the goal is to find $\mathcal{R}^* \subseteq \mathcal{R}$ to cover $S$ such that the \textit{ply} of $\mathcal{R}^*$ is minimized, where the ply is defined as the maximum number of objects in $\mathcal{R}^*$ which have a nonempty common intersection.
Very recently, Durocher et al. gave the first constant-approximation algorithm for MPGSC with unit squares which runs in $O(n^{12})$ time.
We give a significantly simpler constant-approximation algorithm with near-linear running time.
\end{abstract}


\section{Introduction}
Geometric set cover is one of the most fundamental problems in computational geometry \cite{agarwal2014near,aronov2009small,bronnimann1994almost,chan2020faster,clarkson2005improved,varadarajan2010weighted}.
In the original setting of the problem, we are given a set $S$ of points and a set $\mathcal{R}$ of geometric objects, and our goal is to cover all the points in $S$ using fewest objects in $\mathcal{R}$.
Motivated by various applications, several variants of the geometric set cover problem have also been studied in literature.
In this paper, we study a natural variant of the geometric set cover problem, called \textit{minimum-membership geometric set cover} (MMGSC). 

%
%

In the MMGSC problem, the input also consists of a set $S$ of points and a set $\mathcal{R}$ of geometric objects.
Similar to the geometric set cover problem our goal is still to cover all the points in $S$ using the objects in $\mathcal{R}$.
However, we do not care about how many geometric objects we use.
Instead, we want to guarantee that any point in $S$ is not ``over covered''.
More precisely, the goal is to find a subset $\mathcal{R}^* \subseteq \mathcal{R}$ to cover all points in $S$ such that the \textit{membership} of $S$ with respect to $\mathcal{R}^*$, denoted by $\mathsf{memb}(S,\mathcal{R}^*)$, is minimized, where $\mathsf{memb}(S,\mathcal{R}^*) = \max_{p \in S} |\{R \in \mathcal{R}^*: p \in R\}|$.

Kuhn et al.~\cite{KuhnRWWZ05}, motivated by applications in cellular networks, had introduced the non-geometric version of the MMGSC problem, say \textit{minimum-membership set cover} (MMSC).
That is, $S$ is an arbitrary universe with $n$ elements and $\mathcal{R}$ is a collection of subsets of $S$. 
They showed that the MMSC problem admits an $O(\log n)$-approximation algorithm, where $n=|S|$.
Furthermore, they complimented the upper bound result by showing, that unless P$=$NP, the problem cannot be approximated, in polynomial time, by a ratio less than $\ln n$.
Erlebach and van Leeuwen~\cite{ErlebachL08},  in their seminal work on geometric coverage problem, considered the geometric version of MMSC, namely MMGSC,  from the view of approximation algorithms.
They showed  NP-hardness for approximating the problem with ratio less than $2$ on unit disks and unit squares, and gave a $5$-approximation algorithm for unit squares provided that the optimal objective value is bounded by a constant.
More precisely, their algorithm runs in time $n^{O(\mathsf{opt})}$ where $\mathsf{opt}$ is the optimum of the problem (i.e., the minimum membership).
It has remained open that whether MMGSC with unit squares admits a (truly) polynomial-time constant-approximation algorithm.

As our first result, we settle this open question by giving a polynomial-time algorithm for MMGSC with unit squares which achieves a constant approximation ratio.
In fact, our algorithm works for a generalized version of the problem, in which the point set to be covered can be different from the point set whose membership is considered.
\begin{definition}[generalized MMGSC]
In the generalized minimum-membership geometric set cover (MMGSC) problem, the input consists of two sets $S,S'$ of points in $\mathbb{R}^d$ and a set $\mathcal{R}$ of geometric objects in $\mathbb{R}^d$, and the goal is to find a subset $\mathcal{R}^* \subseteq \mathcal{R}$ to cover all points in $S$ such that $\mathsf{memb}(S',\mathcal{R}^*)$ is minimized.
We denote by $\mathsf{opt}(S,S',\mathcal{R})$ the optimum of the problem instance $(S,S',\mathcal{R})$, i.e., $\mathsf{opt}(S,S',\mathcal{R}) = \mathsf{memb}(S',\mathcal{R}^*)$ where $\mathcal{R}^* \subseteq \mathcal{R}$ is an optimal solution.
\end{definition}

\begin{restatable}{theorem}{square} \label{thm-square}
    The generalized MMGSC problem with unit squares admits a polynomial-time constant-approximation algorithm.
\end{restatable}

Our second result is a polynomial-time approximation scheme (PTAS) for MMGSC with halfplanes.
Prior to this work, it was even unknown whether the problem can be approximated in polynomial time with a factor of $o(\log n)$, while the minimum-size set cover problem with halfplanes is polynomial-time solvable.
Again, our PTAS works for the generalized version.

\begin{restatable}{theorem}{half} \label{thm-half}
    The generalized MMGSC problem with halfplanes admits a PTAS.
\end{restatable}

The generalized version of MMGSC is interesting because it also generalizes another closely related problem studied in the literature, called \textit{minimum-ply geometric set cover} (MPGSC).
The MPGSC problems was introduced by Biedl, Biniaz and Lubiw~\cite{BiedlBL21} as a variant of MMGSC.
They observed that in some applications, e.g. interference reduction in cellular networks, it is desirable to minimize the membership of every point in the plane, not only points of $S$.
Therefore, in MPGSC, the goal is to find $\mathcal{R}^* \subseteq \mathcal{R}$ to covers $S$ such that the \textit{ply} of $\mathcal{R}^*$ is minimized, where the ply is defined as the maximum number of objects in $\mathcal{R}^*$ which have a nonempty common intersection.
Observe that MPGSC is a special case of the generalized MMGSC (by letting $S'$ include a point in every face of the arrangement induced by $\mathcal{R}$).
As such, Theorems~\ref{thm-square} and~\ref{thm-half} both apply to MPGSC.

Prior to our work, Biedl, Biniaz and Lubiw~\cite{BiedlBL21} showed that solving the MPGSC with a set of axis-parallel unit squares is NP-hard, and gave a polynomial-time $2$-approximation algorithm for instances in which the optimum (i.e., the minimum ply) is a constant.
Very recently, Durocher, Keil and Mondal~\cite{durocher2023minimum} gave the first constant-approximation algorithm for MPGSC with unit squares, which runs in $O(n^{12})$ time.
This algorithm does not extend to other related settings, such as similarly sized squares or unit disks.
Our algorithm derived from Theorem~\ref{thm-square} is already much more efficient than the one of \cite{durocher2023minimum} (while also not extend to similarly sized squares or unit disks).
However, we observe that for (only) MPGSC  with unit squares, there exists a very simple constant-approximation algorithm which runs in $\widetilde{O}(n)$ time; here $\widetilde{O}$ hides logarithmic factors.
This simple algorithm directly extends to any similarly sized fat objects for which a constant-approximation solution for \textit{minimum-size set cover} can be computed in polynomial time.
Therefore, we obtain the following result.
\begin{restatable}{theorem}{ply}
    The MPGSC problem with unit (or similarly sized) squares/disks admits constant-approximation algorithms with running time $\widetilde{O}(n)$.
\end{restatable}


A common ingredient appearing in all of our results is to establish connections between MMGSC (or MPGSC) and the standard minimum-size geometric set cover.
We show that in certain situations, a minimum-size set cover (satisfying certain conditions) can be a good approximation in terms of MMGSC.
This reveals the underlying relations between different variants of geometric set cover problems, and might be of independent interest.

\paragraph{Other related work.}
Very recently, Mitchell and Pandit \cite{mitchell2021minimum} proved that MMGSC with rectangles intersecting a horizontal line or anchored on two horizontal lines is NP-hard (among other algorithmic and hardness results).

\paragraph{Organization.}
The rest of the paper is organized as follows.
In Section~\ref{sec-unitsq}, we present our result for MMGSC with unit squares.
In Section~\ref{sec-hplane}, we present our result for MMGSC with halfplanes.
The result for MPGSC is given in Section~\ref{sec-MPGSC}.
Finally, in Section~\ref{sec-conclusion}, we conclude the paper and give some open questions for future study.






\section{Constant approximation for unit squares} \label{sec-unitsq}
Let $S,S'$ be two sets of points in $\mathbb{R}^2$ and $\mathcal{Q}$ be a set of (axis-parallel) unit squares.
We want to solve the generalized MMGSC instance $(S,S',\mathcal{Q})$.

\subsection{Restricting $S$ to a grid cell}

First of all, we construct a grid $\varGamma$ consisting of square cells of side-length 1.
For each grid cell $\Box$, we write $S_\Box = S \cap \Box$ and $\mathcal{Q}_\Box = \{Q \in \mathcal{Q}: Q \cap \Box \neq \emptyset\}$.
\begin{lemma} \label{lem-grid}
Suppose that, for every $\Box \in \varGamma$, $\mathcal{Q}_\Box^* \subseteq \mathcal{Q}_\Box$ is a $c$-approximation solution of the generalized MMGSC instance $(S_\Box,S',\mathcal{Q}_\Box)$.
Then $\bigcup_{\Box \in \varGamma} \mathcal{Q}_\Box^*$ is an $O(c)$-approximation solution of the instance $(S,S',\mathcal{Q})$.
\end{lemma}

\begin{proof}
Note that $\bigcup_{\Box \in \varGamma} \mathcal{Q}_\Box^*$ is a set cover of $S$, because any point $p \in S$ is contained in a grid cell $\Box$ and thus $\mathcal{Q}_\Box^*$ covers $p$.
Then we show that for any point $p' \in S'$, the number of unit squares in $\bigcup_{\Box \in \varGamma} \mathcal{Q}_\Box^*$ containing $p'$ is at most $9c \cdot \mathsf{opt}(S,S',\mathcal{Q})$.
Suppose the grid cell containing $p'$ is $\Box'$.
Observe that a unit square $Q \in \bigcup_{\Box \in \varGamma} \mathcal{Q}_\Box^*$ contains $p'$ only if $Q \in \mathcal{Q}_\Box^*$ for a grid cell $\Box$ that is either $\Box'$ or one of the 8 grid cells around $\Box'$.
For each such cell $\Box$, the number of unit squares in $\mathcal{Q}_\Box^*$ containing $p'$ is at most $c \cdot \mathsf{opt}(S_\Box,S',\mathcal{Q}_\Box)$, since $\mathcal{Q}_\Box^*$ is a $c$-approximation solution of $(S_\Box,S',\mathcal{Q}_\Box)$.
Clearly, $\mathsf{opt}(S_\Box,S',\mathcal{Q}_\Box) \leq \mathsf{opt}(S,S',\mathcal{Q})$.
So there can be at most $9c \cdot \mathsf{opt}(S,S',\mathcal{Q})$ unit squares in $\bigcup_{\Box \in \varGamma} \mathcal{Q}_\Box^*$ containing $p'$, which implies that $\bigcup_{\Box \in \varGamma} \mathcal{Q}_\Box^*$ is a $9c$-approximation solution of $(S,S',\mathcal{Q})$.
\end{proof}

\subsection{Partition the instance using LP} \label{sec-LP}
Based on the previous discussion, we will now assume that $S$ is contained in a grid cell $\Box$ and all unit squares in $\mathcal{Q}$ intersect $\Box$.
Note that the points in $S'$ can be everywhere in the plane.
We shall formulate an LP relaxation of the generalized MMGSC instance $(S,S',\mathcal{Q})$.
To this end, we first introduce the notion of fractional set cover.
A \textit{fractional set cover} of a set $A$ of points is a set $\{x_B\}_{B \in \mathcal{B}}$ of numbers in $[0,1]$ indexed by a collection $\mathcal{B}$ of geometric ranges such that $\sum_{B \in \mathcal{B},a \in B} x_B \geq 1$ for all $a \in A$.
For another set $A'$ of points, we can define the \textit{membership} of $A'$ with respect to this fractional set cover $\{x_B\}_{B \in \mathcal{B}}$ as $\mathsf{memb}(A',\{x_B\}_{B \in \mathcal{B}}) = \max_{a' \in A'} \sum_{B \in \mathcal{B}, a' \in B} x_B$.
The LP relaxation of the instance $(S,S',\mathcal{Q})$ simply asks for a fractional set cover of $S$ using the unit squares in $\mathcal{Q}$ that minimizes the membership of $S'$ with respect to it.
Specifically, for each unit square $Q \in \mathcal{Q}$, we create a variable $x_Q$.
In addition, we create another variable $y$, which indicates the upper bound for the membership of $S$ with respect to $\{x_Q\}_{Q \in \mathcal{Q}}$.
We consider the following linear program.
\smallskip
\begin{equation*}
    \begin{array}{rl}
        & \min \text{ } y  \\[2ex]
        \text{s.t. } & 0 \leq x_Q \leq 1 \text{ for all $Q \in \mathcal{Q}$,} \\[1ex]
        & \sum_{Q \in \mathcal{Q}, p \in Q} x_Q \geq 1 \text{ for all $p \in S$,} \\[1ex]
        & \sum_{Q \in \mathcal{Q}, p' \in Q} x_Q \leq y \text{ for all $p' \in S'$.}
    \end{array}
\end{equation*}

\vspace{0.2cm}
We compute an optimal solution $(\{x_Q^*\}_{Q \in \mathcal{Q}},y^*)$ of the above linear program using a polynomial-time LP solver.
We have the following observation about the solution.
\begin{fact}
$y^* \leq \mathsf{opt}(S,S',\mathcal{Q})$.
\end{fact}
\begin{proof}
Let $\mathcal{Q}^* \subseteq \mathcal{Q}$ be an optimal solution.
We have $S \subseteq \bigcup_{Q \in \mathcal{Q}^*}$ and $\mathsf{memb}(S',\mathcal{Q}^*) = \mathsf{opt}(S,S',\mathcal{Q})$.
Set $x_Q= 1$ for $Q \in \mathcal{Q}^*$, $X_Q= 0$ for $Q \notin \mathcal{Q}^*$, and $y = \mathsf{memb}(S',\mathcal{Q}^*)$.
These values satisfy the LP constraints.
Therefore, $y^* \leq y = \mathsf{opt}(S,S',\mathcal{R})$.
\end{proof}

Next, we shall partition the instance $(S,S',\mathcal{Q})$ into four sub-instances according to the LP solution $(\{x_Q^*\}_{Q \in \mathcal{Q}},y^*)$.
Recall that all points in $S$ are inside the grid cell $\Box$ and all unit squares in $\mathcal{Q}$ intersect $\Box$.
Let $c_1,c_2,c_3,c_4$ be the four corners of $\Box$.
We can partition $\mathcal{Q}$ into $\mathcal{Q}_1,\mathcal{Q}_2,\mathcal{Q}_3,\mathcal{Q}_4$, where $\mathcal{Q}_i$ consists of the unit squares containing $c_i$ for $i \in \{1,2,3,4\}$.
Also, we partition $S$ into $S_1,S_2,S_3,S_4$ in the following way.
For a point $p \in \mathbb{R}^2$ and $i \in \{1,2,3,4\}$, define $\delta_{p,i}$ as the sum of $x_Q^*$ for all $Q \in \mathcal{Q}_i$ satisfying $p \in Q$.
Then we assign each point $p \in S$ to $S_i$, where $i \in \{1,2,3,4\}$ is the index that maximizes $\delta_{p,i}$.
Observe the following fact.
\begin{fact}
For each $i \in \{1,2,3,4\}$, we have $\sum_{Q \in \mathcal{Q}_i, p \in Q} x_Q^* \geq \frac{1}{4}$ for all $p \in S_i$.
\end{fact}
\begin{proof}
We have $\sum_{Q \in \mathcal{Q}_i, p \in Q} x_Q^* = \delta_{p,i}$ and $\sum_{i=1}^4 \delta_{p,i} = \sum_{Q \in \mathcal{Q}, p \in Q} x_Q^* \geq 1$, because of the LP constraints.
Furthermore, $\delta_{p,i} \geq \delta_{p,j}$ for all $j \in \{1,2,3,4\}$, as $p \in S_i$.
Thus, $\delta_{p,i} \geq \frac{1}{4}$.
\end{proof}

We now partition the original instance into $(S_1,S',\mathcal{Q}_1),\dots,(S_4,S',\mathcal{Q}_4)$.
Consider an index $i \in \{1,2,3,4\}$.
If we define $\tilde{x}_Q^* = 4 x_Q^*$ for all $Q \in \mathcal{Q}_i$, then the above fact implies $\sum_{Q \in \mathcal{Q}_i, p \in Q} \tilde{x}_Q^* \geq 1$ for all $p \in S_i$.
In other words, $\{\tilde{x}_Q^*\}_{Q \in \mathcal{Q}_i}$ is a fractional set cover of $S_i$.
Note that $\mathsf{memb}(S',\{\tilde{x}_Q^*\}_{Q \in \mathcal{Q}_i}) \leq 4 y^*$, because 
\begin{equation*}
    4y^* \geq \sum_{Q \in \mathcal{Q}, p' \in Q} 4x_Q^* \geq \sum_{Q \in \mathcal{Q}_i, p' \in Q} \tilde{x}_Q^*
\end{equation*}
for all $p' \in S$, due to the constraints of the LP.
With this observation, it now suffices to compute a solution for each instance $(S_i,S',\mathcal{Q}_i)$ that is a constant-factor approximation even with respect to the fractional solutions.
The union of these solutions is a set cover of $S = \bigcup_{i=1}^4 S_i$, the membership of $S'$ with respect to it is $O(y^*)$.
A nice property of the instances $(S_i,S',\mathcal{Q}_i)$ is that all unit squares in $\mathcal{Q}_i$ contain the same corner $c_i$ of $\Box$.
In the next section, we show how to compute the desired approximation solution for such instances.

\subsection{The one-corner case} \label{sec-onecorner}
Now consider an instance $(S,S',\mathcal{Q})$, where all points in $S$ lie in a grid cell $\Box$ and all unit squares contain the same corner (say the bottom-left corner) of $\Box$.
For a point $p \in \mathbb{R}^2$, denote by $x(p)$ and $y(p)$ the $x$-coordinate and $y$-coordinate of $p$, respectively.
Also, for a unit square $Q \in \mathcal{Q}$, denote by $x(Q)$ and $y(Q)$ the $x$-coordinate and $y$-coordinate of the top-right corner of $Q$, respectively.
We make two simple observations.
The first one shows that the integral gap of the minimum-size set cover problem in this setting is equal to 1.
The second one gives a useful geometric property for unit squares containing the same corner of $\Box$.

\begin{fact} \label{fact-nogap}
    Let $S_0 \subseteq S$ be a subset and $\mathcal{Q}_0 \subseteq \mathcal{Q}$ be a minimum-size set cover of $S_0$.
    For any fractional set cover $\{\hat{x}_Q\}_{Q \in \mathcal{Q}}$ of $S_0$, we have $\sum_{Q \in \mathcal{Q}} \hat{x}_Q \geq |\mathcal{Q}_0|$.
\end{fact}
\begin{proof}
For convenience, for a set $A$ of points and a set $\mathcal{B}$ of ranges, we write $\pi(A,\mathcal{B})$ as the minimum size of a set cover of $A$ using the ranges in $\mathcal{B}$.
Now it suffices to show that $\sum_{Q \in \mathcal{Q}} \hat{x}_Q \geq \pi(S_0,\mathcal{Q})$ for any fractional set cover $\{\hat{x}_Q\}_{Q \in \mathcal{Q}}$ of $S_0$.
We apply induction on $\lfloor \sum_{Q \in \mathcal{Q}} \hat{x}_Q \rfloor$.
When $\lfloor \sum_{Q \in \mathcal{Q}} \hat{x}_Q \rfloor = 0$, we must have $S_0 = \emptyset$ and the statement trivially holds.
Suppose the statement holds when $\lfloor \sum_{Q \in \mathcal{Q}} \hat{x}_Q \rfloor \leq k-1$, and consider the case $\lfloor \sum_{Q \in \mathcal{Q}} \hat{x}_Q \rfloor = k$.
Let $p \in S_0$ be the rightmost point in $S_0$, i.e., $x(p) \geq x(q)$ for all $q \in S_0$, and $Q_p \in \mathcal{Q}$ be the unit square containing $p$ that maximizes $y(Q_p)$.
For $Q \in \mathcal{Q}$, we define
\begin{equation*}
    \hat{x}_Q' = \left\{
    \begin{array}{cc}
        \hat{x}_Q & \text{if } p \notin Q, \\
        0 & \text{if } p \in Q \text{ and } Q \neq Q_p, \\
        1 & \text{if } Q = Q_p.
    \end{array}
    \right.
\end{equation*}
Observe that $\sum_{Q \in \mathcal{Q}} \hat{x}_Q' \leq \sum_{Q \in \mathcal{Q}} \hat{x}_Q$.
Indeed, since $\{\hat{x}_Q\}_{Q \in \mathcal{Q}}$ is a fractional set cover of $S_0$, we have $\sum_{Q \in \mathcal{Q},p \in Q} \hat{x}_Q \geq 1 = \sum_{Q \in \mathcal{Q},p \in Q} \hat{x}_Q'$.
Also, $\sum_{Q \in \mathcal{Q}, p \notin Q} \hat{x}_Q' = \sum_{Q \in \mathcal{Q}, p \notin Q} \hat{x}_Q$ by definition.
Thus, $\sum_{Q \in \mathcal{Q}} \hat{x}_Q' \leq \sum_{Q \in \mathcal{Q}} \hat{x}_Q$.
As $\lfloor \sum_{Q \in \mathcal{Q}} \hat{x}_Q \rfloor = k$, we have $\lfloor \sum_{Q \in \mathcal{Q}} \hat{x}_Q' \rfloor \leq k$ and thus $\lfloor \sum_{Q \in \mathcal{Q} \backslash \{Q_p\}} \hat{x}_Q' \rfloor \leq k-1$.
Furthermore, we notice that $\{\hat{x}_Q'\}_{Q \in \mathcal{Q}}$ is also a fractional set cover of $S_0$.
To see this, consider a point $q \in S_0$.
We have $x(q) \leq x(p) \leq x(Q_p)$.
If $q \in Q_p$, then $\sum_{Q \in \mathcal{Q},q \in Q} \hat{x}_Q' \geq 1$.
Assume $q \notin Q_p$, which implies $y(q) > y(Q_p)$ because $x(q) \leq x(Q_p)$.
By the choice of $Q_p$, for all $Q \in \mathcal{Q}$ such that $p \in Q$, we have $y(q) > y(Q)$ and thus $q \notin Q$.
It follows that $\sum_{Q \in \mathcal{Q},q \in Q} \hat{x}_Q' = \sum_{Q \in \mathcal{Q},q \in Q} \hat{x}_Q \geq 1$.
So $\{\hat{x}_Q'\}_{Q \in \mathcal{Q}}$ is a fractional set cover of $S_0$.
Therefore, $\{\hat{x}_Q'\}_{Q \in \mathcal{Q} \backslash \{Q_p\}}$ is a fractional set cover of $S_0 \backslash Q_p$.
Using the fact $\lfloor \sum_{Q \in \mathcal{Q} \backslash \{Q_p\}} \hat{x}_Q' \rfloor \leq k-1$ and our induction hypothesis, we have $\sum_{Q \in \mathcal{Q} \backslash \{Q_p\}} \hat{x}_Q' \geq \pi(S_0 \backslash Q_p,\mathcal{Q} \backslash \{Q_p\})$.
Note that $\pi(S_0 \backslash Q_p,\mathcal{Q} \backslash \{Q_p\}) \geq \pi(S_0,\mathcal{Q}) - 1$, simply because a set cover of $S_0 \backslash Q_p$ together with $Q_p$ forms a set cover of $S_0$.
As a result, $\sum_{Q \in \mathcal{Q} \backslash \{Q_p\}} \hat{x}_Q' \geq \pi(S_0,\mathcal{Q}) - 1$ and hence $\sum_{Q \in \mathcal{Q}} \hat{x}_Q' \geq \pi(S_0,\mathcal{Q})$.
Since $\sum_{Q \in \mathcal{Q}} \hat{x}_Q' \leq \sum_{Q \in \mathcal{Q}} \hat{x}_Q$, we finally have $\sum_{Q \in \mathcal{Q}} \hat{x}_Q \geq \pi(S_0,\mathcal{Q})$, which completes the proof.    
\end{proof}

\begin{fact} \label{fact-sqgeo}
Let $Q^-,Q,Q^+$ be three unit squares all containing the bottom-left corner of $\Box$ which satisfy $x(Q^-) \leq x(Q) \leq x(Q^+)$ and $y(Q^-) \geq y(Q) \geq y(Q^+)$.
Then $Q^- \cap Q^+ \subseteq Q$. 
\end{fact}
\begin{proof}
Let $p \in Q^- \cap Q^+$.
The fact $p \in Q^-$ implies $x(p) \leq x(Q^-)$ and $y(Q^-)-1 \leq y(p)$.
So we have $x(p) \leq x(Q)$ and $y(Q)-1 \leq y(p)$.
On the other hand, the fact $p \in Q^+$ implies $x(Q^+)-1 \leq x(p)$ and $y(p) \leq y(Q^+)$.
So we have $x(Q)-1 \leq x(p)$ and $y(p) \leq y(Q)$.
Therefore, $x(Q)-1 \leq x(p) \leq x(Q)$ and $y(Q)-1 \leq y(p) \leq y(Q)$, which implies that $p \in Q$.
\end{proof}

We say a unit square $Q \in \mathcal{Q}$ is \textit{dominated} by another unit square $Q' \in \mathcal{Q}$ if $Q \cap \Box \subseteq Q' \cap \Box$.
A unit square in $\mathcal{Q}$ is \textit{maximal} if it is not dominated by any other unit squares in $\mathcal{Q}$.
We denote by $\mathcal{Q}_\mathsf{max} \subseteq \mathcal{Q}$ the set of maximal unit squares in $\mathcal{Q}$.
The following lemma shows that any minimum-size set cover of $S$ that only uses the unit squares in $\mathcal{Q}_\mathsf{max}$ is also a good approximation for the minimum-membership set cover.

\begin{lemma}
Let $\mathcal{Q}^* \subseteq \mathcal{Q}$ be a minimum-size set cover of $S$ such that $\mathcal{Q}^* \subseteq \mathcal{Q}_\mathsf{max}$, and $\{x_Q\}_{Q \in \mathcal{Q}}$ be a fractional set cover of $S$.
Then $\mathsf{memb}(S',\mathcal{Q}^*) \leq \mathsf{memb}(S',\{x_Q\}_{Q \in \mathcal{Q}}) +2$.
\end{lemma}
\begin{proof}
Suppose $\mathcal{Q}^* = \{Q_1,\dots,Q_r\}$ where $x(Q_1) \leq \cdots \leq x(Q_r)$.
As $\mathcal{Q}^* \subseteq \mathcal{Q}_\mathsf{max}$, we must have $x(Q_1) < \cdots < x(Q_r)$ and $y(Q_1) > \cdots > y(Q_r)$.
Consider a point $p' \in S'$.
Let $i^- \in [r]$ (resp., $i^+ \in [r]$) be the smallest (resp., largest) index such that $p' \in Q_{i^-}$ (resp., $p' \in Q_{i^+}$).
By Fact~\ref{fact-sqgeo}, we have $p' \in Q_{i^-} \cap Q_{i^+} \subseteq Q_i$ for all $i \in \{i^-,\dots,i^+\}$ and thus $|\{Q \in \mathcal{Q}^*: p' \in Q\}| = i^+ - i^- + 1$.
It suffices to show that $\sum_{Q \in \mathcal{Q}, p' \in Q} x_Q \geq i^+ - i^- - 1$.

Consider the rectangle $R = (x(Q_{i^-}),x(Q_{i^+})] \times (y(Q_{i^+}),y(Q_{i^-})]$; see Figure~\ref{fig-quad}.
Set $S_0 = S \cap R$ and $\mathcal{Q}_0 = \{Q_{i^-+1},\dots,Q_{i^+-1}\}$.
Observe that $\mathcal{Q}_0$ covers $S_0$, since no unit square in $\mathcal{Q}^* \backslash \mathcal{Q}_0$ contains any point in $S_0$.
We claim that $\mathcal{Q}_0 \subseteq \mathcal{Q}$ is a minimum-size set cover of $S_0$.
Indeed, since $x(Q_{i^-}) < \cdots < x(Q_{i^+})$ and $y(Q_{i^-}) > \cdots > y(Q_{i^+})$, the points in $S \backslash S_0$ are all covered by the unit squares $Q_1,\dots,Q_{i^-}$ and $Q_{i^+},\dots,Q_r$.
If there exists a set cover $\mathcal{Q}_0' \subseteq \mathcal{Q}$ of $S_0$ such that $|\mathcal{Q}_0'| < |\mathcal{Q}_0|$, then $\mathcal{Q}_0'$ together with $Q_1,\dots,Q_{i^-},Q_{i^+},\dots,Q_r$ form a set cover of $S$ whose size is smaller than $\mathcal{Q}^*$, contradicting with the fact that $\mathcal{Q}^*$ is a minimum-size set cover of $S$.
Therefore, $\mathcal{Q}_0 \subseteq \mathcal{Q}$ is a minimum-size set cover of $S_0$.
Now for each $Q \in \mathcal{Q}$, define $\hat{x}_Q = x_Q$ if $Q \cap R \neq \emptyset$ and $\hat{x}_Q = 0$ if $Q \cap R = \emptyset$.
As $\{x_Q\}_{Q \in \mathcal{Q}}$ is a fractional set cover of $S$, for each $p \in S_0$, we have $\sum_{Q \in \mathcal{Q},p \in Q} x_Q \geq 1$, which implies $\sum_{Q \in \mathcal{Q},p \in Q} \hat{x}_Q \geq 1$ because $p \in R$ and thus $\hat{x}_Q = x_Q$ for all $Q \in \mathcal{Q}$ such that $p \in Q$.
So $\{\hat{x}_Q\}_{Q \in \mathcal{Q}}$ is a fractional set cover of $S_0$.
By Fact~\ref{fact-nogap}, we then have
\begin{equation} \label{eq-sumhatx1}
    \sum_{Q \in \mathcal{Q}} \hat{x}_Q \geq |\mathcal{Q}_0| = i^+ - i^- - 1.
\end{equation}

\begin{figure}
    \centering
    \includegraphics[height=4cm]{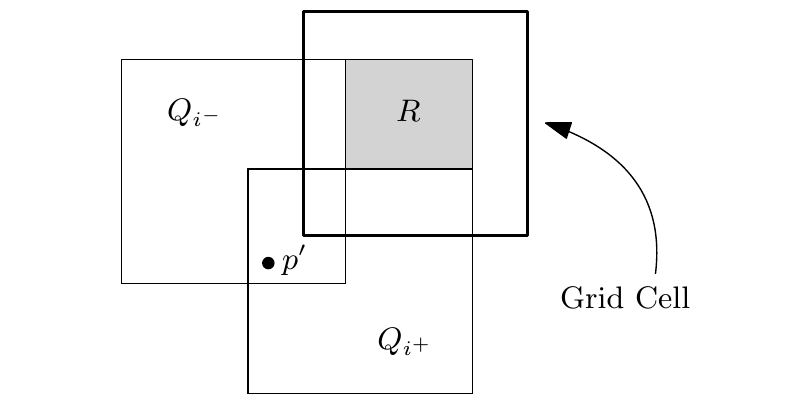}
    \caption{Illustrating the rectangle $R$.}
    \label{fig-quad}
\end{figure}

Next, we observe that $x(Q_{i^-}) \leq x(Q) \leq x(Q_{i^+})$ and $y(Q_{i^-}) \geq y(Q) \geq y(Q_{i^+})$ for any unit square $Q \in \mathcal{Q}$ such that $Q \cap R \neq \emptyset$.
Let $Q \in \mathcal{Q}$ and assume $Q \cap R \neq \emptyset$.
The inequalities $x(Q_{i^-}) \leq x(Q)$ and $y(Q) \geq y(Q_{i^+})$ follow directly from the fact $Q \cap R \neq \emptyset$.
If $x(Q) > x(Q_{i^+})$, then $Q$ dominates $Q_{i^+}$, contradicting the fact $Q_{i^+} \in \mathcal{Q}_\mathsf{max}$.
Similarly, if $y(Q_{i^-}) < y(Q)$, then $Q$ dominates $Q_{i^-}$, contradicting the fact $Q_{i^-} \in \mathcal{Q}_\mathsf{max}$.
Thus, $x(Q_{i^-}) \leq x(Q)$ and $y(Q_{i^-}) \geq y(Q)$.
By Fact~\ref{fact-sqgeo}, we have $p' \in Q_{i^-} \cap Q_{i^+} \subseteq Q$ for all $Q \in \mathcal{Q}$ such that $Q \cap R \neq \emptyset$.
Thus, $\hat{x}_Q = 0$ for all $Q \in \mathcal{Q}$ such that $p' \notin Q$, which implies
\begin{equation} \label{eq-sumhatx2}
    \sum_{Q \in \mathcal{Q},p' \in Q} x_Q \geq \sum_{Q \in \mathcal{Q},p' \in Q} \hat{x}_Q = \sum_{Q \in \mathcal{Q}} \hat{x}_Q.
\end{equation}
Combining Equations~\ref{eq-sumhatx1} and~\ref{eq-sumhatx2}, we have $\sum_{Q \in \mathcal{Q},p' \in Q} x_Q \geq i^+-i^--1$.
\end{proof}

Using the above lemma, now it suffices to compute a minimum-size set cover of $S$ using the unit squares in $\mathcal{Q}_\mathsf{max}$.
It is well-known that in this setting, the minimum-size set cover problem can be solved in polynomial time (or even near-linear time) using a greedy algorithm, because the unit squares in $\mathcal{Q}$ are in fact equivalent to southwest quadrants; see for example \cite{agarwal2022dynamic}.
Thus, we can compute in polynomial time a set cover $\mathcal{Q}^* \subseteq \mathcal{Q}$ of $S$ such that $\mathsf{memb}(S',\mathcal{Q}^*) \leq \mathsf{memb}(S',\{x_Q\}_{Q \in \mathcal{Q}}) +2$ for any fractional set cover $\{x_Q\}_{Q \in \mathcal{Q}}$ of $S$.

\subsection{Putting everything together}
Recall that at the end of Section~\ref{sec-LP}, we have four generalized MMGSC instances $(S_1,S',\mathcal{Q}_1),$ $\dots,(S_4,S',\mathcal{Q}_4)$.
Also, for each $i \in \{1,\dots,4\}$, we have a fractional set cover $\{\tilde{x}_Q^*\}_{Q \in \mathcal{Q}_i}$ of $S_i$ such that $\mathsf{memb}(S',\{\tilde{x}_Q^*\}_{Q \in \mathcal{Q}_i}) \leq 4 y^* \leq 4 \cdot \mathsf{opt}(S,S',\mathcal{Q})$.
By the discussion in Section~\ref{sec-onecorner}, for each $i \in \{1,\dots,4\}$, we can compute in polynomial time a set cover $\mathcal{Q}_i^* \subseteq \mathcal{Q}_i$ of $S_i$ satisfying that $\mathsf{memb}(S',\mathcal{Q}_i^*) \leq \mathsf{memb}(S',\{\tilde{x}_Q^*\}_{Q \in \mathcal{Q}_i}) +2$.
Set $\mathcal{Q}^* = \bigcup_{i=1}^4 \mathcal{Q}_i^*$.
As $S = \bigcup_{i=1}^4 S_i$, $\mathcal{Q}^*$ is a set cover of $S$.
Furthermore, we have
\begin{equation*}
    \begin{array}{rl}
        \mathsf{memb}(S',\mathcal{Q}^*) & \leq \ \sum_{i=1}^4 \mathsf{memb}(S',\mathcal{Q}_i^*) \\[1ex]
        & \leq \ \sum_{i=1}^4 \mathsf{memb}(S',\{\tilde{x}_Q^*\}_{Q \in \mathcal{Q}_i}) +8 \\[1ex]
        & \leq \ 16 y^* + 8 \\[1ex]
        & \leq \ 16 \cdot \mathsf{opt}(S,S',\mathcal{Q}) + 8.
    \end{array}
\end{equation*}
If $\mathsf{opt}(S,S',\mathcal{Q}) > 0$, then $\mathcal{Q}^*$ is a constant-approximation solution.
The case $\mathsf{opt}(S,S',\mathcal{Q}) = 0$ can be easily solved by picking all unit squares in $\mathcal{Q}$ that do not contain any points in $S'$.
Therefore, we obtain a constant-approximation algorithm for the case where $S$ is contained in a grid cell.
Further combining this with Lemma~\ref{lem-grid}, we conclude the following.
\square*

\section{Polynomial-time approximation scheme for halfplanes} \label{sec-hplane}
Let $S,S'$ be two sets of points in $\mathbb{R}^2$ and $\mathcal{H}$ be a set of halfplanes.
We want to solve the generalized MMGSC instance $(S,S',\mathcal{H})$.
Set $n = |S|+|S'|+|\mathcal{H}|$.

In order to describe our algorithm, we first need to introduce some basic notions about halfplanes.
The \textit{normal vector} (or \textit{normal} for short) of a halfplane $H$ is the unit vector perpendicular to the bounding line of $H$ whose direction is to the interior of $H$, that is, if the equation of $H$ is $ax + by + c \geq 0$ where $a^2 + b^2 = 1$, then its normal is $\vec{v} = (a,b)$.
For two nonzero vectors $\vec{u}$ and $\vec{v}$ in the plane, we denote by $\mathsf{ang}(\vec{u},\vec{v})$ the \textit{clockwise ordered angle} from $\vec{u}$ to $\vec{v}$, i.e., the angle between $\vec{u}$ and $\vec{v}$ that is to the clockwise of $\vec{u}$ and to the counter-clockwise of $\vec{v}$.
For two halfplanes $H$ and $J$, we write $\mathsf{ang}(H,J) = \mathsf{ang}(\vec{u},\vec{v})$ where $\vec{u}$ (resp., $\vec{v}$) is the normal of $H$ (resp., $J$).
For a set $\mathcal{R}$ of halfplanes, we use $\bigcap \mathcal{R}$ and $\bigcup \mathcal{R}$ to denote the intersection and the union of all halfplanes in $\mathcal{R}$, respectively.
We say a halfplane $H \in \mathcal{R}$ is \textit{redundant} in $\mathcal{R}$ if $\bigcap \mathcal{R} = \bigcap (\mathcal{R} \backslash \{H\})$.
We say $\mathcal{R}$ is \textit{irreducible} if every halfplane in $\mathcal{R}$ is \textit{not} redundant.
The \textit{complement region} of $\mathcal{R}$ refers to the closure of $\mathbb{R}^2 \backslash \bigcup \mathcal{R}$, which is always a convex polygon (possibly unbounded). 
The following simple facts about halfplanes will be used throughout the section.

\begin{fact} \label{fact-2propR}
    Let $\mathcal{R}$ be an irreducible set of halfplanes such that $\bigcup \mathcal{R} \neq \mathbb{R}^2$.
    Then the following two properties hold.
    \begin{enumerate}[(i)]
        \item For any halfplane $H \in \mathcal{R}$ and another halfplane $H'$ different from $H$, we have that $\bigcup \mathcal{R} \neq \bigcup \mathcal{R}'$, where $\mathcal{R}' = (\mathcal{R} \backslash \{H\}) \cup \{H'\}$.
        \item If the halfplanes in $\mathcal{R}$ has a nonempty intersection, i.e., $\bigcap \mathcal{R} \neq \emptyset$, then we can write $\mathcal{R} = \{H_1,\dots,H_t\}$ such that $0 < \mathsf{ang}(H_1,H_2) < \mathsf{ang}(H_1,H_3) < \cdots < \mathsf{ang}(H_1,H_t) < \pi$.
    \end{enumerate}
\end{fact}
\begin{proof}
We first prove property (i).
As $\mathcal{R}$ is irreducible and $\bigcup \mathcal{R} \neq \mathbb{R}^2$, $H$ contributes a segment (or a ray) $\sigma$ on the boundary of the complement region of $\mathcal{R}$.
If $\bigcup \mathcal{R} = \bigcup \mathcal{R}'$, then $\sigma$ is also a segment on the boundary of the complement region of $\mathcal{R}'$.
Furthermore, the side of $\sigma$ corresponding to the interior of $\bigcup \mathcal{R}$ coincides the side of $\sigma$ corresponding to the interior of $\bigcup \mathcal{R}'$, since $\mathcal{R}$ and $\mathcal{R}'$ share the same complement region.
Therefore, we must have $H \in \mathcal{R}'$ in order to have $\sigma$ on the boundary of the complement region of $\mathcal{R}'$.
This contradicts the assumption $H' \neq H$.

Then we prove property (ii).
Let $p \in \bigcap \mathcal{R}$ and $q \in \mathbb{R}^2 \backslash (\bigcup \mathcal{R})$.
Without loss of generality, assume the line $\overline{pq}$ is vertical and $p$ is below $q$.
Now every halfplane $H \in \mathcal{R}$ contains $p$ but does not contain $q$.
Thus, the equation of $H$ is $y \leq kx + b$ for some $k,b \in \mathbb{R}^2$.
Sort the halfplanes in $\mathcal{R}$ as $H_1,\dots,H_t$ such that if $y \leq k_i x + b_i$ is the equation of $H_i$, then $k_1 \geq \dots \geq k_t$.
Clearly, $0 \leq \mathsf{ang}(H_1,H_2) \leq \mathsf{ang}(H_1,H_3) \leq \cdots \leq \mathsf{ang}(H_1,H_t) < \pi$.
Now observe that $\mathsf{ang}(H_i,H_{i+1}) > 0$.
Indeed, if $\mathsf{ang}(H_i,H_{i+1}) = 0$, then either $H_i \subseteq H_{i+1}$ or $H_{i+1} \subseteq H_i$, which contradicts the assumption that $\mathcal{R} = \{H_1,\dots,H_t\}$ is irreducible.
As such, we have the inequality $0 < \mathsf{ang}(H_1,H_2) < \mathsf{ang}(H_1,H_3) < \cdots < \mathsf{ang}(H_1,H_t) < \pi$.    
\end{proof}

\begin{fact} \label{fact-2propH}
    Let $H_1,\dots,H_t$ be halfplanes such that 
    \begin{equation*}
        0 < \mathsf{ang}(H_1,H_2) < \mathsf{ang}(H_1,H_3) < \cdots < \mathsf{ang}(H_1,H_t) \leq \pi.
    \end{equation*}
    Then the following two properties hold.
    \begin{enumerate}[(i)]
        \item If $H_1 \cup H_t \neq \mathbb{R}^2$, then neither $H_1$ nor $H_t$ is redundant in $\{H_1,\dots,H_t\}$.
        \item If $\{H_1,\dots,H_t\}$ is irreducible, then $\bigcap_{i=1}^t H_i = H_1 \cap H_t$.
    \end{enumerate}
\end{fact}
\begin{proof}
We first prove property (i).
Without loss of generality, it suffices to show $H_1$ is not redundant in $\{H_1,\dots,H_t\}$ and we can assume the equation of $H_1$ is $y \geq 0$.
The boundary line $\ell$ of $H_1$ is defined by the equation $y = 0$.
As $0 < \mathsf{ang}(H_1,H_i) < \pi$ for every $i \in \{2,\dots,t-1\}$, $H_i$ contains the $x=+\infty$ end of $\ell$ but does not contain the $x=-\infty$ end of $\ell$.
If $\mathsf{ang}(H_1,H_t) < \pi$, then $H_t$ also does not contain the $x=-\infty$ end of $\ell$, which implies $\ell \nsubseteq \bigcup_{i=2}^t H_i$ and thus $H_1$ is not redundant in $\{H_1,\dots,H_t\}$.
If $\mathsf{ang}(H_1,H_t) = \pi$, then the equation of $H_t$ is $y \leq c$ for some $c \in \mathbb{R}$.
Note that $c < 0$, as $H_1 \cup H_t \neq \mathbb{R}^2$.
Thus, $H_t$ does not intersect $\ell$ and thus $H_1$ is not redundant in $\{H_1,\dots,H_t\}$.

Then we prove property (ii).
It suffices to show that $H_1 \cap H_t \subseteq H_i$ for all $i \in \{2,\dots,t-1\}$.
Without loss of generality, we only need to consider the case $i=2$.
We first notice that $\mathsf{ang}(H_1,H_2) > 0$.
Indeed, if $\mathsf{ang}(H_1,H_2) = 0$, then either $H_1 \subseteq H_2$ or $H_2 \subseteq H_1$, which contradicts the fact that $\{H_1,\dots,H_t\}$ is irreducible.
For the same reason, $\mathsf{ang}(H_2,H_t) > 0$.
Thus, $\mathsf{ang}(H_1,H_2) < \mathsf{ang}(H_1,H_t)$.
If $\mathsf{ang}(H_1,H_t) = \pi$, then $H_1$ and $H_t$ are halfplanes with opposite direction, which implies either $H_1 \cap H_t = \emptyset$ or $H_1 \cup H_t = \mathbb{R}^2$.
In the former case, we directly have $H_1 \cap H_t \subseteq H_2$.
In the latter case, we have $t = 2$ (and thus $H_1 \cap H_t \subseteq H_2$ trivially holds), for otherwise $H_2,\dots,H_{t-1}$ are all redundant in $\{H_1,\dots,H_t\}$, contradicting the assumption that $\{H_1,\dots,H_t\}$ is irreducible.
So it suffice to consider the case $\mathsf{ang}(H_1,H_t) < \pi$.
Now $\mathsf{ang}(H_1,H_2) < \mathsf{ang}(H_1,H_t) < \pi$.
Without loss of generality, assume the equation of $H_2$ is $x \geq 0$.
Then the equation of $H_1$ is $y \geq kx+b$ for some $k,b \in \mathbb{R}$ and the equation of $H_t$ is $y \leq k'x+b'$ for some $k',b' \in \mathbb{R}$.
Let $p$ be the intersection of the lines $y = kx+b$ and $y = k'x+b'$.
Note that $H_1 \cup H_t$ covers all points with $x$-coordinates greater than or equal to $p$.
It follows that the $x$ coordinate of $p$ is positive, for otherwise $H_2 \subseteq H_1 \cup H_t$, which contradicts the fact that $\{H_1,\dots,H_t\}$ is irreducible.
Also notice that the area $H_1 \cap H_t$ is entirely to the right of $p$, and is thus contained in $H_2$.    
\end{proof}

\subsection{An $n^{O(\mathsf{opt})}$-time exact algorithm} \label{sec-n^k}
In this section, we show how to compute an (exact) optimal solution of the instance $(S,S',\mathcal{H})$ in $n^{O(\mathsf{opt})}$ time.
It suffices to solve a decision problem: given an integer $k \geq 0$, find a subset $\mathcal{Z} \subseteq \mathcal{H}$ which covers $S$ and satisfies $\mathsf{memb}(S',\mathcal{Z}) \leq k$ or decide that such a subset does not exist.
As long as this problem can be solved in $n^{O(k)}$ time, by trying $k = 1, \dots, |\mathcal{H}|$, we can finally compute an optimal solution of $(S,S',\mathcal{H})$ in $n^{O(\mathsf{opt})}$ time.
In what follows, a \textit{valid} solution of $(S,S',\mathcal{H})$ refers to a subset $\mathcal{Z} \subseteq \mathcal{H}$ which covers $S$ and satisfies $\mathsf{memb}(S',\mathcal{Z}) \leq k$.

Let $\Delta$ be a sufficiently large number such that $S \cup S' \subseteq [-\Delta,\Delta]^2$.
For convenience, we add to $\mathcal{H}$ four dummy halfplanes with equations $y \leq -\Delta$, $y \geq \Delta$, $x \leq -\Delta$, and $x \geq \Delta$.
As these dummy halfplanes does not contain any points in $S \cup S'$, including them in $\mathcal{H}$ does not change the problem.
We say a set of halfplanes is \textit{regular} if it is irreducible and its complement region is nonempty and bounded.
We have the following simple observation.
\begin{fact} \label{fact-regular}
If $(S,S',\mathcal{H})$ has a valid solution, then either it has a regular valid solution or it has a valid solution that covers the entire plane $\mathbb{R}^2$.
\end{fact}
\begin{proof}
Consider a valid solution $\mathcal{Z} \subseteq \mathcal{H}$ of $(S,S',\mathcal{H})$.
We shall create another valid solution which either is regular or covers $\mathbb{R}^2$.
First, we add the four dummy halfplanes to $\mathcal{Z}$ and let $\mathcal{Z}' \supseteq \mathcal{Z}$ denote the resulting set.
Then $\mathcal{Z}'$ covers $S$ and $\mathsf{memb}(S',\mathcal{Z}') = \mathsf{memb}(S',\mathcal{Z})$. 
So $\mathcal{Z}'$ is also a valid solution.
Next, we obtain an irreducible subset $\mathcal{Z}'' \subseteq \mathcal{Z}'$ which has the same complement region as $\mathcal{Z}'$ by keeping removing redundant halfplanes from $\mathcal{Z}'$.
Now $\mathcal{Z}''$ still covers $S$ and $\mathsf{memb}(S',\mathcal{Z}'') \leq \mathsf{memb}(S',\mathcal{Z}')$, which implies that $\mathcal{Z}''$ is also a valid solution.
If $\bigcup \mathcal{Z}'' = \mathbb{R}^2$, we are done.
Otherwise, the complement region of $\mathcal{Z}''$ is nonempty.
Furthermore, the complement region of $\mathcal{Z}''$ is the same as that of $\mathcal{Z}'$, where the latter is bounded because the four dummy halfplanes are all in $\mathcal{Z}'$.
Therefore, the complement region of $\mathcal{Z}''$ is bounded.
This implies $\mathcal{Z}''$ is regular, because $\bigcup \mathcal{Z}''$ is irreducible.
\end{proof}

If $(S,S',\mathcal{H})$ has a valid solution that covers $\mathbb{R}^2$, then it also has an irreducible valid solution that covers $\mathbb{R}^2$, which is of size at most $3$ by Helly's theorem.
Therefore, in this case, we can solve the problem in $n^{O(1)}$ time by simply enumerating all subsets of $\mathcal{H}$ of size at most $3$.
Otherwise, by the above fact, it suffices to check whether there exists a regular valid solution of $(S,S',\mathcal{H})$.
In what follows, we assume $(S,S',\mathcal{H})$ has a regular valid solution and show how to find such a solution in $n^{O(k)}$ time.
If our algorithm does not find a regular valid solution at the end, we can conclude its non-existence.
Let $\mathcal{Z} \subseteq \mathcal{H}$ be a (unknown) regular valid solution of $(S,S',\mathcal{H})$.
By definition, the complement region of $\mathcal{Z}$ is nonempty, and is a (bounded) convex polygon.
Consider the arrangement $\mathcal{A}$ of the boundary lines of the halfplanes in $\mathcal{H}$.
This arrangement has $O(n^2)$ faces, among which at least one face is contained in the complement region of $\mathcal{Z}$.
We simply guess such a face.
By making $O(n^2)$ guesses, we can assume that we know a face $F$ in the complement region of $\mathcal{Z}$.
Then we take a point $p$ in the interior of $F$, which is also in the interior of the complement region of $\mathcal{Z}$.

Now the problem becomes finding a regular valid solution of $(S,S',\mathcal{H})$ whose complement region contains $p$.
Therefore, we can remove from $\mathcal{H}$ all halfplanes that contain $p$.
Now the complement region of any subset of $\mathcal{H}$ contains $p$.
We say a convex polygon $\varGamma$ is \textit{$\mathcal{H}$-compatible} if each edge $e$ of $\varGamma$ is a portion of the boundary line of some halfplane $H \in \mathcal{H}$ such that $\varGamma \cap H = e$ (or equivalently $H$ does not contain $\varGamma$).
Note that the complement region of a regular valid solution is an $\mathcal{H}$-compatible convex polygon $\varGamma$ which satisfies \textbf{(i)} no point in $S$ lies in the interior of $\varGamma$ and \textbf{(ii)} for any $k+1$ edges $e_1,\dots,e_{k+1}$ of $\varGamma$, the intersection $\bigcap_{i=1}^{k+1} H(e_i)$ does not contain any point in $S'$; here $H(e) \in \mathcal{H}$ denotes the halfplane whose boundary line containing $e$ and $\varGamma \cap H = e$.
On the other hand, every $\mathcal{H}$-compatible convex polygon satisfying conditions \textbf{(i)} and \textbf{(ii)} is the complement region of a regular valid solution of $(S,S',\mathcal{H})$, which is just the set of halfplanes corresponding to the edges of $\varGamma$.
With this observation, it suffices to find an $\mathcal{H}$-compatible convex polygon $\varGamma$ satisfying the two conditions.
The follow fact can be used to simplify condition \textbf{(ii)}.

\begin{fact}
If there exist $t$ edges $e_1,\dots,e_t$ of an $\mathcal{H}$-compatible convex polygon $\varGamma$ with $(\bigcap_{i=1}^t H(e_i)) \cap S' \neq \emptyset$, then there exist $t$ consecutive edges $f_1,\dots,f_t$ of $\varGamma$ with $(\bigcap_{i=1}^t H(f_i)) \cap S' \neq \emptyset$.
\end{fact}
\begin{proof}
Set $\mathcal{R} = \{H(e_1),\dots,H(e_t)\}$ and assume $(\bigcap \mathcal{R}) \cap S' \neq \emptyset$, which implies $\bigcap \mathcal{R} \neq \emptyset$.
Note that the set of halfplanes corresponding to the edges of $\varGamma$ are irreducible, because $\varGamma$ is $\mathcal{H}$-compatible and thus the interior of each edge $e$ of $\varGamma$ can only be covered by the halfplane $H(e)$.
In particular, $\mathcal{R}$ is irreducible.
Also, $\bigcup \mathcal{R} \neq \mathbb{R}^2$ by our assumption $\bigcup \mathcal{H} \neq \mathbb{R}^2$.
Therefore, by (ii) of Fact~\ref{fact-2propR}, there exist $e^-,e^+\in \{e_1,\dots,e_t\}$ such that $\mathsf{ang}(H(e^-),H(e_i)) \leq \mathsf{ang}(H(e^-),H(e^+)) < \pi$ for all $i \in [t]$.
Now we go clockwise along the boundary of $\varGamma$ from $e^-$ to $e^+$, and let $E$ be the set of edges of $\varGamma$ we visit (including $e^-$ and $e^+$).
Clearly, $e_i \in E$ for all $i \in [t]$ and thus $|E| \geq t$.
Furthermore, $0 < \mathsf{ang}(H(e^-),H(e)) < \mathsf{ang}(H(e^-),H(e^+)) < \pi$ for all $e \in E \backslash \{e^-,e^+\}$.
Define $\mathcal{R}' = \{H(e): e \in E\}$.
Since $\mathcal{R}$ and $\mathcal{R}'$ are both irreducible, we can apply (ii) of Fact~\ref{fact-2propH} to deduce $\bigcap \mathcal{R} = H(e^-) \cap H(e^+) = \bigcap \mathcal{R}'$, which implies $(\bigcap \mathcal{R}') \cap S' \neq \emptyset$.
Finally, because $E$ consists of consecutive edges of $\varGamma$ and $|E| \geq t$, there exist $f_1,\dots,f_t \in E$ which are $t$ consecutive edges of $\varGamma$.
We have $\bigcap \mathcal{R}' \subseteq \bigcap_{i=1}^t H(f_i)$ and thus $(\bigcap_{i=1}^t H(f_i)) \cap S' \neq \emptyset$.
\end{proof}

By the above fact, we only need to find an $\mathcal{H}$-compatible convex polygon $\varGamma$ which satisfies \textbf{(i)} no point in $S$ lies in the interior of $\varGamma$ and \textbf{(ii)} $(\bigcap_{i=1}^{k+1} H(e_i)) \cap S' \neq \emptyset$ for any $k+1$ \textit{consecutive} edges $e_1,\dots,e_{k+1}$ of $\varGamma$.
For convenience, we say $\varGamma$ is \textit{well-behaved} if it satisfies conditions \textbf{(i)} and \textbf{(ii)}.
Next, we reduce this problem to a shortest-cycle problem in a (weighted) directed graph $G$ as follows.
Let $\mathcal{L}$ denote the set of boundary lines of halfplanes in $\mathcal{H}$.
We consider every segment $s$ in the plane which is on some line $\ell \in \mathcal{L}$ and satisfies that each endpoint of $s$ is the intersection point of $\ell$ and another line in $\mathcal{L}$.
We use $\varPhi$ to denote the set of these segments.
Note that $|\varPhi| = O(n^3)$, as $\varPhi$ contains $O(n^2)$ segments on each line $\ell \in \mathcal{L}$.
Clearly, the edges of an $\mathcal{H}$-compatible convex polygon are all segments in $\varPhi$.
Consider a segment $\phi \in \varPhi$.
Recall that the point $p$ is the interior of $F$, which is a face of the arrangement $\mathcal{A}$.
Thus, no line in $\mathcal{L}$ goes through $p$.
It follows that for every segment $\phi \in \varPhi$, the two endpoints of $\phi$ and $p$ form a triangle $\Delta_\phi$.
If $p \rightarrow a \rightarrow b$ is the clockwise ordering of the three vertices of $\Delta_\phi$ from $p$, then we call $a$ the \textit{left} endpoint of $\phi$ and call $b$ the \textit{right} endpoint of $\phi$.
Clearly, $\mathsf{ang}(\overrightarrow{pa},\overrightarrow{pb}) < \pi$.
The vertices of the graph $G$ to be constructed are one-to-one corresponding to the $(k+1)$-tuples $(\phi_0,\phi_1,\dots,\phi_k) \in \varPhi^{k+1}$ which satisfy the following three conditions.

\smallskip
\begin{enumerate}[~1.]
    \item The left endpoint of $\phi_i$ is the right endpoint of $\phi_{i-1}$ for all $i \in [k]$.
    Below we use $a_i$ to denote the left endpoint of $\phi_i$ (i.e., the right endpoint of $\phi_{i-1}$).
    This condition guarantees that the segments $\phi_0,\phi_1,\dots,\phi_k$ form a polygonal chain of $k+1$ pieces.
    \item $\mathsf{ang}(\overrightarrow{a_{i-1}a_i},\overrightarrow{a_ia_{i+1}}) \leq \pi$ for all $i \in [k]$.
    This condition guarantees that the chain formed by $\phi_0,\phi_1,\dots,\phi_k$ is \textit{clockwise convex}, in the sense that when we go along the chain from $a_0$ to $a_{k+1}$, we always turn \textit{right} at the vertices of the chain.
    Figure~\ref{fig-chain} shows a chain satisfying this condition (and also condition~1).
    \item $S \cap (\bigcup_{i=0}^k \Delta_{\phi_i}) \subseteq \bigcup_{i=0}^k \phi_i$ and $(\bigcap_{i=0}^k H(\phi_i)) \cap S' = \emptyset$.
\end{enumerate}
\smallskip

\begin{figure}
    \centering
    \includegraphics[height=5.5cm]{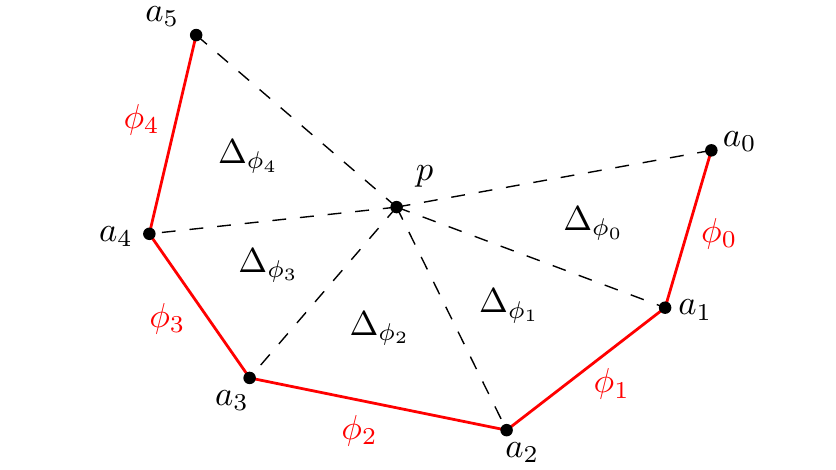}
    \caption{Illustrating the conditions for a vertex of $G$.}
    \label{fig-chain}
\end{figure}

Intuitively, the $(k+1)$-tuple corresponding to each vertex of $G$ represents a possible choice of $k+1$ consecutive edges of the $\mathcal{H}$-compatible convex polygon we are looking for.
For two vertices $v = (\phi_0,\phi_1,\dots,\phi_k)$ and $v' = (\phi_0',\phi_1',\dots,\phi_k')$ such that $(\phi_1,\dots,\phi_k) = (\phi_0',\phi_1',\dots,\phi_{k-1}')$, we add a directed edge from $v$ to $v'$ with weight $\mathsf{ang}(\overrightarrow{pa},\overrightarrow{pb})$, where $a$ is the left endpoint of $\phi_0$ and $b$ is the left endpoint of $\phi_0' = \phi_1$ (which is also the right endpoint of $\phi_0$ by condition~1 above).
Note that the weight of every edge of $G$ is positive, since the two endpoints of every $\phi \in \varPhi$ and $p$ form a triangle $\Delta_\phi$.
The key observation is the following lemma.
\begin{lemma} \label{lem-shortcycle}
There exists a well-behaved $\mathcal{H}$-compatible convex polygon containing $p$ iff the (weighted) length of a shortest cycle in $G$ is exactly $2\pi$.
\end{lemma}
\begin{proof}
We first observe that the (weighted) length of any cycle in $G$ is a multiple of $2 \pi$ and is at least $2 \pi$.
For a vertex $v = (\phi_0,\phi_1,\dots,\phi_k)$ of $G$, we denote by $a_v$ the left endpoint of $\phi_0$.
Consider a cycle $(v_0,v_1,\dots,v_r)$ in $G$ where $v_r = v_0$.
By the construction of $G$, the length of the cycle $(v_0,\dots,v_r)$ is $\sum_{i=0}^{r-1} \mathsf{ang}(\overrightarrow{pa_{v_i}},\overrightarrow{pa_{v_{i+1}}})$, which is a multiple of $2 \pi$ since $\overrightarrow{pa_{v_0}} = \overrightarrow{pa_{v_r}}$.
Furthermore, as the weight of every edge of $G$ is positive, the length of the cycle is greater than 0 and thus is at least $2 \pi$.
With this observation, it now suffices to show that there exists a well-behaved $\mathcal{H}$-compatible convex polygon containing $p$ iff there exists a cycle in $G$ of length exactly $2\pi$.

For the ``if'' part, assume there exists a cycle $(v_0,v_1,\dots,v_r)$ in $G$ of length exactly $2\pi$, where $v_r = v_0$.
By the construction of $G$, there exist $\phi_0,\phi_1,\dots,\phi_{r-1} \in \varPhi$ such that $v_i = (\phi_i,\dots,\phi_{i+k})$ for any $i \geq 0$; here for convenience, we define $\phi_i = \phi_{i \text{ mod } r}$ for any $i \geq r$.
Let $a_i$ denote the left endpoint of $\phi_i$.
As the length of $(v_0,v_1,\dots,v_r)$ is $2 \pi$, we have $\sum_{i=0}^{r-1} \mathsf{ang}(\overrightarrow{pa_i},\overrightarrow{pa_{i+1}}) = 2 \pi$.
Therefore, the points $a_0,a_1,\dots,a_{r-1}$ are located around $p$ in clockwise order, and the segments $\phi_0,\phi_1,\dots,\phi_{r-1}$ enclose a simple polygon $\varGamma$ containing $p$.
We have $\varGamma = \bigcup_{i=0}^{r-1} \Delta_{\phi_i}$.
Recall the three conditions for defining the vertices of $G$.
As $(\phi_i,\dots,\phi_{i+k})$ is a vertex of $G$ for every $i$, condition~2 guarantees the convexity of $\varGamma$.
Furthermore, since $p \notin H(\phi_i)$, which implies that $\varGamma$ and $H(\phi_i)$ are on the opposite sides of $\phi_i$, i.e., $\varGamma \cap H(\phi_i) = \phi_i$.
Therefore, $\varGamma$ is an $\mathcal{H}$-compatible convex polygon.
It suffices to verify that $\varGamma$ is well-behaved.
Note that $S \cap \varGamma = S \cap (\bigcup_{i=0}^{r-1} \Delta_{\phi_i}) = \bigcup_{i=0}^{r-1} (S \cap (\bigcup_{j=i}^{i+k} \Delta_{\phi_j}))$.
By condition~3, we have $S \cap (\bigcup_{j=i}^{i+k} \Delta_{\phi_j}) \subseteq \bigcup_{j=i}^{i+k} \phi_j$ for any $i$, which implies $S \cap \varGamma \subseteq \bigcup_{i=0}^{r-1} \phi_i$, i.e., no point in $S$ lies in the interior of $\varGamma$.
Again by condition~3, $(\bigcup_{j=i}^{i+k} H(\phi_j)) \cap S' = \emptyset$ for any $i$.
Thus, $\varGamma$ is well-behaved.

For the ``only if'' part, assume there exists a well-behaved $\mathcal{H}$-compatible convex polygon $\varGamma$ containing $p$.
Suppose $\phi_0,\phi_1,\dots,\phi_{r-1}$ are the edges of $\varGamma$ in clockwise order.
These edges are segments in $\varPhi$ since $\varGamma$ is $\mathcal{H}$-compatible.
Same as before, we define $\phi_i = \phi_{i \text{ mod } r}$ for any $i \geq r$.
Now for any $i \geq 0$, the $(k+1)$-tuple $(\phi_i,\dots,\phi_{i+k})$ represents $k+1$ consecutive edges of $\varGamma$.
Based on the fact that $\varGamma$ is a well-behaved $\mathcal{H}$-compatible convex polygon containing $p$, one can easily check that $(\phi_i,\dots,\phi_{i+k})$ satisfies conditions 1, 2, 3 for defining the vertices of $G$.
First, condition~1 clearly holds.
By the convexity of $\varGamma$, we then have condition~2.
To see condition~3, note that $\varGamma$ is $\mathcal{H}$-compatible, and thus the halfplanes $H(\phi_i),\dots,H(\phi_{i+k})$ do not intersect the interior of $\varGamma$.
But $\bigcup_{j=i}^{i+k} \Delta_{\phi_j}$ is contained in $\varGamma$.
Because $\varGamma$ is well-behaved, every point in $S \cap \varGamma$ lies on the boundary of $\varGamma$, and thus $S \cap \bigcup_{j=i}^{i+k} \Delta_{\phi_j} \subseteq \bigcup_{j=i}^{i+k} \phi_j$.
The emptiness of $(\bigcap_{j=i}^{i+k} H(\phi_j)) \cap S'$ also follows from the fact that $\varGamma$ is well-behaved.
Therefore, condition~3 holds.
It follows that $(\phi_i,\dots,\phi_{i+k})$ is a vertex of $G$, which we denote by $v_i$.
By the construction of $G$, there is an edge between $v_i$ and $v_{i+1}$ for any $i \geq 0$.
Since $v_r = v_0$, $(v_0,v_1,\dots,v_r)$ is a cycle in $G$.
If we use $a_i$ to denote the left endpoint of $\phi_i$ (which is also the right endpoint of $\phi_{i-1}$), then $a_0,a_1,\dots,a_{r-1}$ are the vertices of $\varGamma$ in clockwise order.
The length of the cycle $(v_0,v_1,\dots,v_r)$ is then $\sum_{i=0}^{r-1} \mathsf{ang}(\overrightarrow{pa_i},\overrightarrow{pa_{i+1}}) = 2 \pi$.
\end{proof}

Based on the above lemma, it suffices to compute a shortest cycle in $G$, which can be done by standard algorithms (e.g., Dijkstra) in polynomial time in the size of $G$.
Note that $G$ has $n^{O(k)}$ vertices.
Therefore, we obtain an $n^{O(k)}$-time algorithm for computing a set cover $\mathcal{Z} \subseteq \mathcal{H}$ of $S$ satisfying $\mathsf{memb}(S',\mathcal{Z}) \leq k$, if such a set cover exists.
By iteratively trying $k = 1,\dots,|\mathcal{H}|$, we can solve the MMGSC problem with halfplanes in $n^{O(\mathsf{opt})}$ time.

\subsection{An algorithm with constant additive error} \label{sec-adderr}
In this section, we show how to compute in polynomial time an approximation solution $\mathcal{Z} \subseteq \mathcal{H}$ of the instance $(S,S',\mathcal{H})$ with constant additive error, that is, $\mathsf{memb}(S',\mathcal{Z}) = \mathsf{opt}(S,S',\mathcal{H})+ O(1)$.
If $\bigcup \mathcal{H} = \mathbb{R}^2$, then by Helly's theorem, there exist $H_1,H_2,H_3 \in \mathcal{H}$ such that $H_1 \cup H_2 \cup H_3 = \mathbb{R}^2$.
In this case, we can take $\{H_1,H_2,H_3\}$ as our solution, which clearly has constant additive error.
So assume $\bigcup \mathcal{H} \neq \mathbb{R}^2$.
Our algorithm is in the spirit of local search.
However, different from most local-search algorithms which improve the ``quality'' of the solution in each step (via local modifications), our algorithm does not care about the quality (i.e., membership), and instead focuses on shrinking the complement region of the solution.
Formally, for two sets $\mathcal{Z}$ and $\mathcal{Z}'$, we write $\mathcal{Z} \prec \mathcal{Z}'$ if $\bigcup \mathcal{Z} \subsetneq \bigcup \mathcal{Z}'$, and $\mathcal{Z} \preceq \mathcal{Z}'$ if $\bigcup \mathcal{Z} \subseteq \bigcup \mathcal{Z}'$.
We define the following notion of ``locally (non-)improvable'' solutions.

\begin{definition}
A subset $\mathcal{Z} \subseteq \mathcal{H}$ is \textbf{$k$-expandable} if there exists $\mathcal{Z}' \subseteq \mathcal{H}$ such that $|\mathcal{Z} \backslash \mathcal{Z}'| = |\mathcal{Z}' \backslash \mathcal{Z}| \leq k$ and $\mathcal{Z} \prec \mathcal{Z}'$.
A subset of $\mathcal{H}$ is \textbf{$k$-stable} if it is not $k$-expandable.
\end{definition}

In other words, $\mathcal{Z} \subseteq \mathcal{H}$ is $k$-expandable (resp., $k$-stable) if we can (resp., cannot) replace $k$ halfplanes in $\mathcal{Z}$ with other $k$ halfplanes in $\mathcal{H}$ to shrink the complement region of $\mathcal{Z}$.
We are interested in subsets $\mathcal{Z} \subseteq \mathcal{H}$ that are \textit{minimum-size} set covers of $S$ and are $k$-stable.
Such a set can be constructed via the standard local-search procedure.

\begin{lemma} \label{lem-localsearch}
A minimum-size set cover $\mathcal{Z} \subseteq \mathcal{H}$ of $S$ that is $k$-stable can be computed in $n^{O(k)}$ time.
\end{lemma}
\begin{proof}
The standard set cover problem for halfplanes is polynomial-time solvable.
So we can compute a minimum-size set cover $\mathcal{Z} \subseteq \mathcal{H}$ of $S$ in $n^{O(1)}$ time.
To further obtain a $k$-stable one, we keep doing the following procedure.
Whenever there exists $\mathcal{Z}' \subseteq \mathcal{H}$ such that $|\mathcal{Z} \backslash \mathcal{Z}'| = |\mathcal{Z}' \backslash \mathcal{Z}| \leq k$ and $\mathcal{Z} \prec \mathcal{Z}'$, we update $\mathcal{Z}$ to $\mathcal{Z}'$.
During this procedure, the size of $\mathcal{Z}$ does not change and the complement region of $\mathcal{Z}$ shrinks.
So $\mathcal{Z}$ is always a minimum-size set cover of $S$.
Furthermore, as the complement region of $\mathcal{Z}$ shrinks in every step, the procedure will finally terminate.
At the end, $\mathcal{Z}$ is not $k$-expandable and is thus $k$-stable.
This proves the correctness of our algorithm.
To see it takes $n^{O(k)}$ time, we show that \textbf{(i)} we terminate in $O(n)$ steps and \textbf{(ii)} each step can be implemented in $n^{O(k)}$ time.

For \textbf{(i)}, the key observation is that every halfplane $H \in \mathcal{H}$ can be removed from $\mathcal{Z}$ at most once.
Formally, we denote by $\mathcal{Z}_i$ the set $\mathcal{Z}$ after the $i$-th step of the procedure, and thus the original $\mathcal{Z}$ is $\mathcal{Z}_0$.
Let $P_i$ be the complement region of $\mathcal{Z}_i$.
Suppose $H \in \mathcal{Z}_{i-1}$ and $H \notin \mathcal{Z}_i$.
We claim that $H \notin \mathcal{Z}_j$ for all $j > i$.
Assume $H \in \mathcal{Z}_j$ for some $j > i$.
Since $\mathcal{Z}_j$ is a minimum-size set cover of $S$, $H$ is not redundant in $\mathcal{Z}_j$ and thus one edge $e$ of $P_j$ is defined by $H$, i.e., $e$ is a segment on the boundary line of $H$.
Note that $e$ is also a portion of the boundary of $P_{i-1}$, because $H \in \mathcal{Z}_{i-1}$ and $P_j \subseteq P_{i-1}$.
It follows that $e$ is a portion of the boundary of $P_i$, since $P_j \subseteq P_i \subseteq P_{i-1}$.
But this cannot be the case, as $H \notin \mathcal{Z}_i$.
Thus, $H \notin \mathcal{Z}_j$ for all $j > i$.
Now for every index $i \geq 1$, there exists at least one halfplane $H \in \mathcal{H}$ such that $H \in \mathcal{Z}_{i-1}$ and $H \notin \mathcal{Z}_i$, simply because $|\mathcal{Z}_{i-1}| = |\mathcal{Z}_i|$ and $\mathcal{Z}_{i-1} \neq \mathcal{Z}_i$.
We then charge the $i$-th step to this halfplane $H$.
By the above observation, each halfplane is charged at most once.
Therefore, the procedure terminates in at most $n$ steps.
To see \textbf{(ii)}, observe that in each step, the number of $\mathcal{Z}' \subseteq \mathcal{H}$ satisfying $|\mathcal{Z} \backslash \mathcal{Z}'| = |\mathcal{Z}' \backslash \mathcal{Z}| \leq k$ is bounded by $n^{O(k)}$, and these sets can be enumerated in $n^{O(k)}$ time.
So each step can be implemented in $n^{O(k)}$ time.
As a result, the entire algorithm terminates in $n^{O(k)}$ time.
\end{proof}

Our key observation is that any minimum-size set cover of $S$ that is $k$-stable has additive error at most $2$ in terms of MMGSC, even for $k = 1$.

\begin{lemma} \label{lem-adderr}
If $\mathcal{Z} \subseteq \mathcal{H}$ is a minimum-size set cover of $S$ that is $1$-stable, then we have $|\mathcal{Z}| \leq \mathsf{opt}(S,S',\mathcal{H}) + 2$.
\end{lemma}
\begin{proof}
Consider a point $p \in S'$.
We show that $\mathsf{memb}(p,\mathcal{Z}') \geq \mathsf{memb}(p,\mathcal{Z})-2$ for any set cover $\mathcal{Z}' \subseteq \mathcal{H}$ of $S$.
Let $\mathcal{Z}(p) \subseteq \mathcal{Z}$ consist of all halfplanes in $\mathcal{Z}$ that contain $p$.
As $\bigcap \mathcal{Z}(p) \neq \emptyset$, by (ii) of Fact~\ref{fact-2propR} and the assumption $\bigcup \mathcal{H} \neq \mathbb{R}^2$, we have $\mathcal{Z}(p) = \{H_1,\dots,H_r\}$ such that $0 < \mathsf{ang}(H_1,H_2) < \mathsf{ang}(H_1,H_3) < \cdots < \mathsf{ang}(H_1,H_r) < \pi$.
Let $S_0 \subseteq S$ consist of points contained in $\bigcup_{i=2}^{r-1} H_i$ but not contained in any other halfplanes in $\mathcal{Z}$, and $\mathcal{Z}_0' \subseteq \mathcal{Z}'$ consist of halfplanes that contain at least one point in $S_0$.
Note that $|\mathcal{Z}_0'| \geq r-2$, for otherwise $(\mathcal{Z} \backslash \{H_2,\dots,H_{r-1}\}) \cup \mathcal{Z}_0'$ is a set cover of $S$ of size strictly smaller than $\mathcal{Z}$, which contradicts the fact that $\mathcal{Z}$ is a minimum-size set cover of $S$.
We shall show that every halfplane in $\mathcal{Z}_0'$ contains $p$ and thus
\begin{equation*}
    \mathsf{memb}(p,\mathcal{Z}') \geq \mathsf{memb}(p,\mathcal{Z}_0') = |\mathcal{Z}_0'| \geq r-2 = \mathsf{memb}(p,\mathcal{Z})-2.
\end{equation*}
Consider a halfplane $H' \in \mathcal{Z}_0'$.
We want to show $p \in H'$.
By the construction of $\mathcal{Z}_0'$, $H'$ contains a point $q \in S_0$.
Furthermore, by the construction of $S_0$, $q$ is contained in $\bigcup_{i=2}^{r-1} H_i$ but not contained in any halfplane in $\mathcal{Z} \backslash \{H_2,\dots,H_{r-1}\}$.
In particular, $q \notin H_1$ and $q \notin H_r$, which implies $H' \neq H_1$ and $H' \neq H_r$.
We observe that $\{H_1,H_r,H'\}$ is irreducible.
Clearly, $H'$ is not redundant in $\{H_1,H_r,H'\}$, as it contains $q$ while $H_1$ and $H_r$ do not contain $q$.
If $H_1$ is redundant in $\{H_1,H_r,H'\}$, then $\mathcal{Z} \preceq (\mathcal{Z} \backslash \{H_1\}) \cup \{H'\}$.
Since $\mathcal{Z}$ is irreducible and $H' \neq H_1$, by (i) of Fact~\ref{fact-2propR}, this implies $\mathcal{Z} \prec (\mathcal{Z} \backslash \{H_1\}) \cup \{H'\}$, which contradicts the fact that $\mathcal{Z}$ is $1$-stable.
So $H_1$ is not redundant in $\{H_1,H_r,H'\}$.
For the same reason, $H_r$ is also not redundant in $\{H_1,H_r,H'\}$.
Thus, $\{H_1,H_r,H'\}$ is irreducible.

In what follows, we complete the proof by showing that either $p \in H'$ or $\mathcal{Z}$ is $1$-expandable.
As the latter is false (for $\mathcal{Z}$ is $1$-stable), this implies $p \in H'$.
If $\mathsf{ang}(H_1,H') < \mathsf{ang}(H_1,H_r)$, by the irreducibility of $\{H_1,H_r,H'\}$ and (ii) of Fact~\ref{fact-2propH}, we have $H_1 \cap H' \cap H_r = H_1 \cap H_r$, which implies $H_1 \cap H_r \subseteq H'$ and thus $p \in H_1 \cap H_r \subseteq H'$.
If $\mathsf{ang}(H_1,H') = \mathsf{ang}(H_1,H_r)$, then either $H' \subseteq H_r$ or $H_r \subseteq H'$.
Note that the former is not true as $q \in H'$ but $q \notin H_r$.
Thus, we have $p \in H_r \subseteq H'$.
It suffices to consider the case $\mathsf{ang}(H_1,H') > \mathsf{ang}(H_1,H_r)$.
In this case, we show that $\mathcal{Z}$ is $1$-expandable.
Since $q \in \bigcup_{i=2}^{r-1} H_i$, there exists $H \in \{H_2,\dots,H_{r-1}\}$ which contains $q$.
Now $\mathsf{ang}(H_1,H) < \mathsf{ang}(H_1,H_r) < \mathsf{ang}(H_1,H')$, which implies $\mathsf{ang}(H,H_r) < \mathsf{ang}(H,H')$ and $\mathsf{ang}(H',H_1) < \mathsf{ang}(H',H)$.
We further distinguish two cases, $\mathsf{ang}(H,H') \leq \pi$ and $\mathsf{ang}(H,H') \geq \pi$ (which are in fact symmetric).
Assume $\mathsf{ang}(H,H') \leq \pi$.
Figure~\ref{fig-4half} shows the situation of the points $p,q$ and the halfplanes $H,H',H_1,H_r$ this case.
As $\mathsf{ang}(H,H_r) < \mathsf{ang}(H,H')$, by (ii) of Fact~\ref{fact-2propH}, if $\{H,H_r,H'\}$ is irreducible, then $H \cap H_r \cap H' = H \cap H'$.
But $H \cap H_r \cap H' \neq H \cap H'$, because $q \in H \cap H'$ and $q \notin H_r$.
Thus, $\{H,H_r,H'\}$ is reducible.
Note that $H \cup H' \neq \mathbb{R}^2$, since $\bigcup \mathcal{H} \neq \mathbb{R}^2$ by our assumption.
So by (i) of Fact~\ref{fact-2propH}, neither $H$ nor $H'$ is redundant in $\{H,H_r,H'\}$.
It follows that $H_r$ is redundant in $\{H,H_r,H'\}$, because $\{H,H_r,H'\}$ is reducible.
Therefore, $\mathcal{Z} \preceq (\mathcal{Z} \backslash \{H_r\}) \cup \{H'\}$.
Since $\mathcal{Z}$ is irreducible and $H' \neq H_r$, by (i) of Fact~\ref{fact-2propR}, we have $\mathcal{Z} \prec (\mathcal{Z} \backslash \{H_r\}) \cup \{H'\}$, i.e., $\mathcal{Z}$ is $1$-expandable.
The other case $\mathsf{ang}(H,H') \geq \pi$ is similar.
In this case, $\mathsf{ang}(H',H) \leq \pi$.
Using the fact $\mathsf{ang}(H',H_1) < \mathsf{ang}(H',H)$ and the same argument as above, we can show that $\mathcal{Z} \prec (\mathcal{Z} \backslash \{H_1\}) \cup \{H'\}$, i.e., $\mathcal{Z}$ is $1$-expandable.
\end{proof}

\begin{figure}[h]
    \centering
    \includegraphics[height=6.5cm]{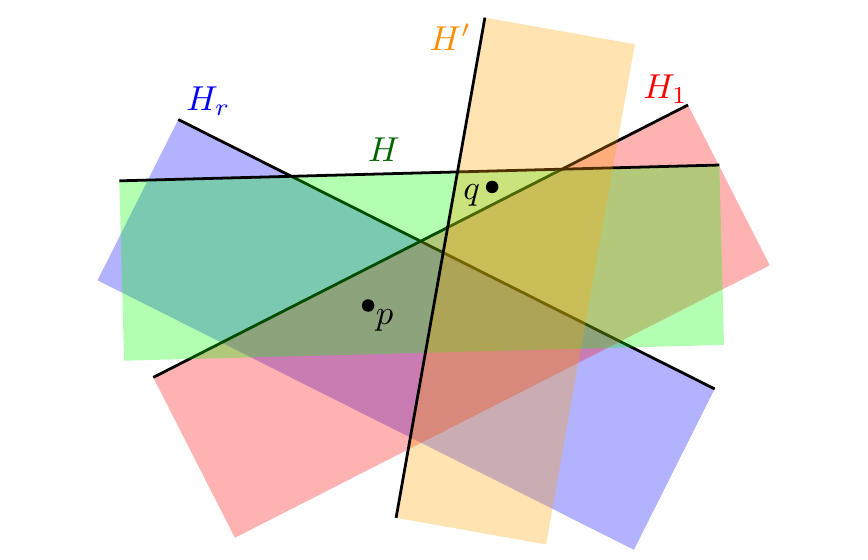}
    \caption{Illustration of the proof of Lemma~\ref{lem-adderr}.}
    \label{fig-4half}
\end{figure}

Using Lemma~\ref{lem-localsearch}, we can compute a $1$-stable minimum-size set cover $\mathcal{Z} \subseteq \mathcal{H}$ of $S$ in $n^{O(1)}$ time.
Then by Lemma~\ref{lem-adderr}, $\mathcal{Z}$ is an approximation solution for the MMGSC instance $(S,S',\mathcal{H})$ with additive error $2$.
This gives us a polynomial-time approximation algorithm for MMGSC with halfplanes with $O(1)$ additive error.

\subsection{Putting everything together}

Our PTAS can be obtained by directly combining the algorithms in Sections~\ref{sec-n^k} and~\ref{sec-adderr}.
Let $c = O(1)$ be the additive error of the algorithm in Section~\ref{sec-adderr}.
We first run the algorithm in Section~\ref{sec-adderr} to obtain a solution $\mathcal{Z} \subseteq \mathcal{H}$.
If $|\mathcal{Z}| \geq \frac{1+\varepsilon}{\varepsilon} \cdot c$, then
\begin{equation*}
    \frac{|\mathcal{Z}|}{\mathsf{opt}(S,S',\mathcal{H})} \leq \frac{|\mathcal{Z}|}{|\mathcal{Z}|-c} \leq 1+\varepsilon.
\end{equation*}
In this case, $\mathcal{Z}$ is already a $(1+\varepsilon)$-approximation solution.
Otherwise, $|\mathcal{Z}| < \frac{1+\varepsilon}{\varepsilon} \cdot c$ and thus $\mathsf{opt}(S,S',\mathcal{H}) < \frac{1+\varepsilon}{\varepsilon} \cdot c$.
We can then run the algorithm in Section~\ref{sec-n^k} to compute an optimal solution in $n^{O(1/\varepsilon)}$ time.
So we conclude the following.
\half*

\section{Minimum-ply geometric set cover} \label{sec-MPGSC}
In this section, we give a very simple constant-approximation algorithm for minimum-ply geometric set cover with unit squares.
The technique can be applied to the problem with any similarly-sized geometric objects in $\mathbb{R}^2$.

Let $(S,\mathcal{Q})$ be an MPGSC instance.
As in Section~\ref{sec-unitsq}, we first apply the grid techinique.
We construct a grid $\varGamma$ consisting of square cells of side-length 1.
For each grid cell $\Box$, write $S_\Box = S \cap \Box$ and $\mathcal{Q}_\Box = \{Q \in \mathcal{Q}: Q \cap \Box \neq \emptyset\}$.
The key observation is the following.
\begin{lemma} \label{lem-decomp}
Suppose that, for every $\Box \in \varGamma$, $\mathcal{Q}_\Box^* \subseteq \mathcal{Q}_\Box$ is a $c$-approximation solution of the minimum-size geometric set cover instance $(S_\Box,\mathcal{Q}_\Box)$.
Then $\bigcup_{\Box \in \varGamma} \mathcal{Q}_\Box^*$ is an $O(c)$-approximation solution of the MPGSC instance $(S,\mathcal{Q})$.
\end{lemma}
\begin{proof}
Let $\gamma = \mathsf{ply}(\bigcup_{\Box \in \varGamma} \mathcal{Q}_\Box^*)$ and $p \in \mathbb{R}^2$ be a point contained in $\gamma$ unit squares in $\bigcup_{\Box \in \varGamma} \mathcal{Q}_\Box^*$.
Consider the grid cell $\Box_p$ containing $p$ and define $\mathcal{C}$ as the set of $3 \times 3$ grid cells centered at $\Box_p$.
Note that all unit squares containing $p$ belong to $\bigcup_{\Box \in \mathcal{C}} \mathcal{Q}_\Box^*$.
So we have $|\bigcup_{\Box \in \mathcal{C}} \mathcal{Q}_\Box^*| \geq \gamma$ and $|\max_{\Box \in \mathcal{C}} \mathcal{Q}_\Box^*| \geq \gamma/9$.
Therefore, there exists $\Box \in \varGamma$ such that $|\mathcal{Q}_\Box^*| \geq \gamma/9$.
As $\mathcal{Q}_\Box^*$ is a $c$-approximation solution of the minimum-size set cover instance $(S_\Box,\mathcal{Q}_\Box)$, we know that any subset of $\mathcal{Q}_\Box$ that covers $S_\Box$ has size at least $\gamma/(9c)$.
It follows that any subset of $\mathcal{Q}$ that covers $S$ must include at least $\gamma/(9c)$ unit squares in $\mathcal{Q}_\Box$.
Note that each of these unit squares contains a corner of $\Box$.
Thus, at least one corner of $\Box$ is contained in $\gamma/(36c)$ such unit squares, which implies that the ply of any solution is at least $\gamma/(36c)$.
As a result, $\bigcup_{\Box \in \varGamma} \mathcal{Q}_\Box^*$ is an $O(c)$-approximation solution of the MPGSC instance $(S,\mathcal{Q})$.
\end{proof}

Note that the argument in the above proof can be extended to any similarly-size fat objects in any fixed dimension.
Here a set of geometric objects are \textit{similarly-size fat objects} if there exist constants $\alpha,\beta > 0$ such that every object in the set contains a ball of radius $\alpha$ and is contained in a ball of radius $\beta$.

\begin{theorem} \label{thm-reduction}
    For any class $\mathcal{C}$ of similarly sized fat objects in $\mathbb{R}^d$, if the minimum-size geometric set cover problem with $\mathcal{C}$ admits a constant-approximation algorithm with running time $T(n)$ for a function $T$ satisfying $T(a+b) \geq T(a)+T(b)$, then the MPGSC problem with $\mathcal{C}$ also admits a constant-approximation algorithm with running time $T(n)$.
\end{theorem}
\begin{proof}
Let $(S,\mathcal{R})$ be an MPGSC instance where $\mathcal{R} \subseteq \mathcal{C}$.
We use the above grid technique to decompose the input instance $(S,\mathcal{R})$ into a set $\{(S_\Box,\mathcal{R}_\Box)\}$ of instances.
Then apply the algorithm for minimum-size geometric set cover problem with $\mathcal{C}$ to compute constant-approximation (with respect to size) solutions $\mathcal{R}_\Box^* \subseteq \mathcal{R}_\Box$.
By Lemma~\ref{lem-decomp}, $\bigcup_{\Box \in \varGamma} \mathcal{R}_\Box^*$ is a constant-approximation solution of the MPGSC instance $(S,\mathcal{R})$.
If the algorithm for minimum-size set cover runs in $T(n)$ time, then our algorithm also takes $T(n)$ time, as long as the function $T$ satisfies $T(a+b) \geq T(a)+T(b)$.
\end{proof}

\ply*
\begin{proof}
The $\widetilde{O}(n)$-time constant-approximation algorithms for minimum-size set cover with similarly sized squares/disks are well-known.
For similarly sized squares, see for example~\cite{agarwal2022dynamic}.
For similarly sized disks, see for example~\cite{agarwal2014near,chan2020faster}.
Applying Theorem~\ref{thm-reduction} directly yields $\widetilde{O}(n)$-time constant-approximation algorithms for MPGSC with unit squares and unit disks.
\end{proof}

\section{Conclusion} \label{sec-conclusion}
In this paper, we revisited natural variants of the classical geometric set cover problem, the \textit{minimum-membership geometric set cover} (MMGSC) problem and the \textit{minimum-ply geometric set cover} (MPGSC).
We presented a polynomial-time constant-approximation algorithm for MMGSC with unit squares, a PTAS for MMGSC with halfplanes, and a simple constant-approximation algorithm for MPGSC with unit squares (or unit disks) that runs in near-linear time.
Our algorithms is to establish connections between MMGSC/MPGSC and the standard minimum-size geometric set cover, which might be of independent interest.

Next, we give some open questions for future study on this topic.
First, our algorithm for MMGSC with unit squares heavily relies on the geometry of \textit{unit squares}, and thus cannot be generalized to other settings such as similarly sized squares or unit disks (while our algorithm for MPGSC can).
So it is interesting to consider the problem with these ranges.
Second, it is still not clear whether the MMGSC problem with halfplanes is NP-hard or not.
Thus, an important question is to settle the complexity of this problem.
Finally, the classical geometric set cover problem has been recently studied in dynamic settings \cite{agarwal2022dynamic,chan2022more,chan2022dynamic,khan2023online}, where points/ranges can be inserted and deleted.
Naturally, one can also study the dynamic version of MMGSC and MPGSC.

\section*{Acknowledgements}
The authors would like to thank Qizheng He, Daniel Lokshtanov, Rahul Saladi, Subhash Suri, and Haitao Wang for helpful discussions about the problems, and thank the anonymous reviewers for their detailed comments, which help significantly improve the writing of the paper.

\bibliographystyle{plainurl}
\bibliography{my_bib}

\end{document}